\def\UseBibLatex{1}

\ifx\AlgorithmicaVer\undefined%
  \newcommand{\AlgVer}[1]{}
  \newcommand{\NotAlgVer}[1]{#1}
\else
  \newcommand{\AlgVer}[1]{#1}
  \newcommand{\NotAlgVer}[1]{}
\fi

\makeatletter
\def\input@path{{styles/}{../styles/}}
\makeatother

\AlgVer{%
   \documentclass[pdflatex,sn-mathphys]{sn-jnl}
}
\NotAlgVer{%
   \documentclass[12pt]{article}
}

\usepackage[cm]{fullpage}%

\usepackage{graphicx}

\usepackage{amsmath}%
\usepackage{amssymb}%
\usepackage{xcolor}%
\usepackage{euscript}%
\usepackage{caption}

\usepackage{xcolor}%
\usepackage{mleftright}%
\usepackage{xspace}%
\usepackage{paralist}%

\usepackage{hyperref}%
\hypersetup{%
  breaklinks,%
  ocgcolorlinks, colorlinks=true,%
  urlcolor=[rgb]{0.25,0.0,0.0},%
  linkcolor=[rgb]{0.5,0.0,0.0},%
  citecolor=[rgb]{0,0.2,0.445},%
  filecolor=[rgb]{0,0,0.4},
  anchorcolor=[rgb]={0.0,0.1,0.2}%
}

\NotAlgVer{%
   \usepackage[amsmath,thmmarks]{ntheorem}%
   \theoremseparator{.}%
}

\numberwithin{figure}{section}%
\numberwithin{table}{section}%
\numberwithin{equation}{section}%

\AlgVer{%
   \jyear{2021}%

   \theoremstyle{thmstyleone}%
   \newtheorem{theorem}{Theorem}%
   \newtheorem{lemma}[theorem]{Lemma}%

   \theoremstyle{thmstyletwo}%
   \newtheorem{example}[theorem]{Example}%
   \newtheorem{remark}[theorem]{Remark}%
   
   \theoremstyle{thmstyletwo}%
   \newtheorem{problem}[theorem]{Problem}%
}

\providecommand{\BibLatexMode}[1]{}
\providecommand{\BibTexMode}[1]{#1}

\ifx\UseBibLatex\undefined%
  \renewcommand{\BibLatexMode}[1]{}
  \renewcommand{\BibTexMode}[1]{#1}
\else
  \renewcommand{\BibLatexMode}[1]{#1}
  \renewcommand{\BibTexMode}[1]{}
\fi

\BibLatexMode{%
   \usepackage{sariel_biblatex}%
   \usepackage[english]{babel}%
   \usepackage{csquotes}
   \renewcommand*{\multicitedelim}{\addcomma\space}%
}

\NotAlgVer{%
\theoremseparator{.}%

\theoremstyle{plain}%
\newtheorem{theorem}{Theorem}[section]

\newtheorem{lemma}[theorem]{Lemma}

\theoremstyle{plain}%
\theoremheaderfont{\sf} \theorembodyfont{\upshape}%
\newtheorem*{remark:unnumbered}[theorem]{Remark}%
\newtheorem{remark}[theorem]{Remark}%

\newtheorem{example}[theorem]{Example}

\newtheorem{problem}[theorem]{Problem}

\newcommand{\myqedsymbol}{\rule{2mm}{2mm}}

\theoremheaderfont{\em}%
\theorembodyfont{\upshape}%
\theoremstyle{nonumberplain}%
\theoremseparator{}%
\theoremsymbol{\myqedsymbol}%
\newtheorem{proof}{Proof:}%

}

\newcommand{\HLink}[2]{\hyperref[#2]{#1~\ref*{#2}}}
\newcommand{\HLinkSuffix}[3]{\hyperref[#2]{#1\ref*{#2}{#3}}}

\newcommand{\figlab}[1]{\label{fig:#1}}
\newcommand{\figref}[1]{\HLink{Figure}{fig:#1}}

\newcommand{\thmlab}[1]{{\label{theo:#1}}}
\newcommand{\thmref}[1]{\HLink{Theorem}{theo:#1}}

\newcommand{\seclab}[1]{\label{sec:#1}}
\newcommand{\secref}[1]{\HLink{Section}{sec:#1}}

\newcommand{\problab}[1]{\label{prob:#1}}
\newcommand{\probref}[1]{\HLink{Problem}{prob:#1}}%

\newcommand{\lemlab}[1]{\label{lemma:#1}}
\newcommand{\lemref}[1]{\HLink{Lemma}{lemma:#1}}%

\providecommand{\eqlab}[1]{}%
\renewcommand{\eqlab}[1]{\label{equation:#1}}

\DeclareFontFamily{U}{BOONDOX-calo}{\skewchar\font=45 }
\DeclareFontShape{U}{BOONDOX-calo}{m}{n}{
  <-> s*[1.05] BOONDOX-r-calo}{}
\DeclareFontShape{U}{BOONDOX-calo}{b}{n}{
  <-> s*[1.05] BOONDOX-b-calo}{}
\DeclareMathAlphabet{\mathcalb}{U}{BOONDOX-calo}{m}{n}
\SetMathAlphabet{\mathcalb}{bold}{U}{BOONDOX-calo}{b}{n}
\DeclareMathAlphabet{\mathbcalb}{U}{BOONDOX-calo}{b}{n}

\providecommand{\remove}[1]{}%
\newcommand{\Set}[2]{\left\{ #1 \;\middle\vert\; #2 \right\}}
\newcommand{\pth}[2][\!]{\mleft({#2}\mright)}%

\newcommand{\brc}[1]{\left\{ {#1} \right\}}
\newcommand{\cardin}[1]{\left\vert {#1}  \right\vert}

\renewcommand{\th}{th\xspace}

\renewcommand{\Re}{\mathbb{R}}%

\newcommand{\USet}{U}%

\newcommand{\GSet}{G}%

\newcommand{\GSSet}{\mathcal{G}}%

\newcommand{\HS}{\Mh{\Pi}}%
\newcommand{\HSA}{\mathcal{Y}}%
\newcommand{\HSB}{\mathcal{Z}}%

\newcommand{\HSS}{\mathcal{S}}%

\newcommand{\curr}{\mathcal{C}}%

\newcommand{\he}{\Mh{\mathcalb{e}}}%
\newcommand{\heo}{\Mh{\mathcalb{o}}}%
\newcommand{\hy}{\Mh{\mathcalb{y}}}%
\newcommand{\hz}{\Mh{\mathcalb{z}}}%

\newcommand{\PTD}{\textsf{PTD}\xspace}

\newcommand{\ZZ}{\mathbb{Z}}%
\newcommand{\margY}[2]{\Delta^{}_{#2}\pth{#1}}%
\newcommand{\fmax}{f_{\max}}

\newcommand{\demandChar}{\Mh{\mathsf{d}}}%
\newcommand{\demandX}[1]{\demandChar\pth{ #1 }}%
\newcommand{\demandY}[2]{\demandChar_{#1}\pth{ #2 }}%

\newcommand{\faceY}[2]{\Mh{\mathrm{face}}\pth{ #1, #2}}%
\renewcommand{\th}{th\xspace}
\newcommand{\Arr}{\Mh{\mathop{\mathrm{\EuScript{A}}}}}%
\newcommand{\ArrX}[1]{\Arr\pth{#1}}%
\newcommand{\cell}{\Mh{\psi}}%

\newcommand{\pA}{\Mh{p}}%
\newcommand{\pB}{\Mh{q}}%

\providecommand{\Mh}[1]{#1}%

\newcommand{\pa}{\Mh{x}}%
\providecommand{\P}{\Mh{P}}%
\renewcommand{\P}{\Mh{P}}%

\newcommand{\Opt}{\mathcal{O}}
\newcommand{\kopt}{k}

\newcommand{\PCMS}{\textsf{PCMS}\xspace}
\newcommand{\cutX}[1]{\mathrm{cut}\pth{#1}}%
\newcommand{\RMC}{\textsf{RMC}\xspace}%

\newlength{\savedparindent}

\definecolor{almostblack}{rgb}{0, 0, 0.3}
\newcommand{\emphw}[1]{{\textcolor{almostblack}{\emph{#1}}}}%

\definecolor{blue25}{rgb}{0,0,0.7}
\providecommand{\emphic}[2]{%
   \textcolor{blue25}{%
      \textbf{\emph{#1}}}%
   \index{#2}}
\providecommand{\emphi}[1]{\emphic{#1}{#1}}

\newcommand{\F}{\Mh{\mathcal{F}}}%

\NotAlgVer{%

\newcommand{\atgen}{\symbol{'100}}
\newcommand{\SarielThanks}[1]{\thanks{Department of Computer Science;
      University of Illinois; 201 N. Goodwin Avenue; Urbana, IL,
      61801, USA; {\tt sariel\atgen{}illinois.edu}; {\tt
         \url{http://sarielhp.org/}.} #1}}

\newcommand{\MitchellThanks}[1]{%
   \thanks{%
      Department of Computer Science;
      University of Illinois; 201 N. Goodwin Avenue; Urbana, IL,
      61801, USA; {\tt mfjones2\atgen{}illinois.edu}; {\tt
         \url{http://mfjones2.web.engr.illinois.edu/}.} #1}}

}
\BibLatexMode{%
   \bibliography{split_splat}%
}

\begin{document}

\NotAlgVer{%
   \title{Few Cuts Meet Many Point Sets%
      \thanks{Sariel Har-Peled partially supported by NSF AF awards
         CCF-1421231, CCF-1217462, and CCF-1907400. Mitchell Jones
         partially supported by NSF AF awards CCF-1421231 and
         CCF-1907400.}%
   }

   \author{%
      Sariel Har-Peled\SarielThanks{}%
      \and%
      Mitchell Jones\MitchellThanks{}%
   } }

\AlgVer{%
   \institute{%
      S. Har-Peled \at %
      Department of Computer Science; University of Illinois; %
      201 N. Goodwin Avenue; Urbana, IL, 61801, USA. %
      \\
      \email{sariel@illinois.edu}%
      \and%
      M. Jones \at %
      Department of Computer Science; %
      University of Illinois; %
      201 N. Goodwin Avenue; %
      Urbana, IL, 61801, USA. \\ %
      \email{mfjones2@illinois.edu}%
   }
   
}

\AlgVer{%
\title[Few Cuts Meet Many Point Sets]{Few Cuts Meet Many Point Sets%
   }%

\author[1]{\fnm{Sariel} \sur{Har-Peled}} \email{sariel@illinois.edu}
\author[1]{\fnm{Mitchell} \sur{Jones}}
\email{mitchell.jones1994@gmail.com}

\affil*[1]{%
   \orgdiv{Department of Computer Science}, %
   \orgname{University of Illinois Urbana-Champaign}, %
   \orgaddress{%
      \street{201 N. Goodwin Avenue}, %
      \city{Urbana}, %
      \postcode{61801}, %
      \state{IL}, %
      \country{USA}}%
}

}

\NotAlgVer{\maketitle}

\abstract{%
   We study the problem of how to split many point sets in $\Re^d$
   into smaller parts using a few (shared) splitting hyperplanes. This
   problem is related to the classical Ham-Sandwich Theorem. We
   provide a logarithmic approximation to the optimal solution using
   the greedy algorithm for submodular optimization.  }

\AlgVer{%
   \keywords{Ham-Sandwich Theorem, Submodular optimization, Approximation
      algorithms}%
}

\AlgVer{\maketitle}

\section{Introduction}

\subsection{Motivation \& the problem}

A basic problem in algorithms is partitioning the data effectively, so
that one can apply divide and conquer algorithms. Recently, there was
significant progress \cite{mp-mppsr-15, aaez-eagpp-21, s-pmit-22} on
using polynomials to perform such partitions (e.g., polynomial
Ham-Sandwich Theorem) to derive better combinatorial bounds (and in
some cases, algorithms).  Thus, polynomials provide a ``universal''
solution to this problem -- however, there are some technical
difficulties in handling polynomials efficiently. This work deals with
alternative partitioning geometric schemes using lines or hyperplanes,
and figuring out how one can do it efficiently.

\paragraph{Example: Separating points by a polynomial.}
As a concrete example, consider the problem of splitting a point set
$\P \subset \Re^2$ into singletons. This requires computing a non-zero
polynomial $p(x,y)$, with a zero set
$Z = \Set{ (x,y) \in \Re^2 }{p(x,y) = 0}$, such that for every point
of $\P \setminus Z$ lies in its own connected component of
$\Re^2 \setminus Z$.

Such a polynomial can be computed using the polynomial Ham-Sandwich
theorem. At the $i$\th stage, the point set is partitioned into $2^i$
sets $\P^i_1, \ldots, \P^i_{2^i}$ of similar cardinality. The idea is
now to lift the points of $\P$ into $2^{i}$ dimensions. To this end,
let $\F(i)$ be the set of the first $i$ monomials over $x$ and $y$
ordered by their degrees (i.e.,
$\F(i) = \{ x, y, xy,x^2, y^2, xy, x^3, \cdots \}$). One then map a
point $(x,y) \in \Re^2$, to the corresponding point
$(x,y, xy,x^2,y^2, x^3, \cdots)$ in $2^i$ dimensions, where each
coordinate is a monomial from the set $\F(i)$. In the lifted space,
one can now halve all $2^i$ sets by a single hyperplane, as guaranteed
by the Ham-Sandwich Theorem, which in the original plane corresponds
to a polynomial. This breaks $\P$ into $2^{i+1}$ sets, and one
continues to the next iteration. If $f_i$ is the polynomial computed
in the $i$\th iteration, for $i=1, \ldots, h = \log n$, then the zero
set $Z_f$ of the product polynomial $f(x,y) = \prod_{i=1}^h f_i(x,y)$
breaks the plane into the desired components, as can be easily
verified.

\paragraph{Why partitioning by polynomials is sometime not
   sufficient.}

The main issue is that the zero sets of polynomials are not easy to
manipulate. If one preserves the representation of $f$ as a product
polynomial, as described above, then it is easy to decide if two
points are in the same connected component of $\Re^2 \setminus
Z_f$. However, this task becomes much harder if the polynomial is not
provided in this form. Furthermore, this representation is not easy to
modify and adapt (for example, modifying the representation if a few
more points are inserted).  As mentioned above, a natural alternative
is to separate points by lines (or hyperplanes in higher
dimensions). Here, two points $\pA, \pB$ are separated by a given set
of lines if there is at least one line in the set that intersects the
interior of the segment $\pA \pB$.

\paragraph{The specific problem: Halving point sets.}
The input is made out of $m$ sets $\P_1, \ldots, \P_m$ of points in
$\Re^d$, not necessarily disjoint (with $m > d$). Our goal is to split
these sets into equal parts using a minimal number of hyperplanes. For
$m\leq d$, the Ham-Sandwich Theorem states that one can bisect all of
the sets using a single hyperplane. However, for $m > d$ and
non-degenerate inputs, this is no longer possible. In particular, the
number of point sets $m$ might be significantly larger than $d$. One
way to get around this restriction is via the polynomial Ham-Sandwich
Theorem \cite{st-gst-42}, as described above.

Here, we are interested in what can be done with restricted entities,
such as (several) hyperplanes.  To keep the problem feasible, we
somewhat relax the problem---the requirement is no longer that each
piece of $\P_i$ is exactly half the size of the original set, but
rather that it is sufficiently small.

\begin{figure}[t]
    \centerline{\includegraphics{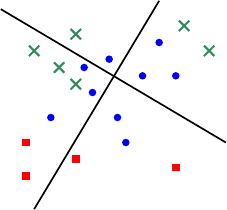}}%
    \captionsetup{width=.8\linewidth}
    \caption{Given three point sets, suppose the goal is to break the
       green (cross) point set into sets with at most three points,
       the blue (dot) point set into sets with at most four points,
       and the red (square) point set into sets with at most two
       points. This can be achieved using two separating lines.}
    \figlab{example}
\end{figure}

\begin{problem}
    \problab{prob}%
    Let $\P_1, \ldots, \P_m$ be $m > d$ point sets in $\Re^d$, not
    necessarily disjoint, with $n = \sum_i \cardin{ \P_i}$. Let
    $\mu_1, \ldots, \mu_m$ be integers with
    $0 < \mu_i \leq \cardin{\P_i}$. The goal is to compute the
    smallest set of hyperplanes $H$, such that for every cell $\cell$
    in the arrangement $\ArrX{H}$ of hyperplanes,
    $\cardin{\P_i \cap \cell} = \cardin{\Set{p \in \P_i}{p \in c}}
    \leq \mu_i$ for all $i$.  See \figref{example} for an example.
\end{problem}

This problem is interesting even for $d=2$, $m=1$, and $\mu_1=1$ --
this is the problem of breaking a set of points in the plane into
singletons using lines. Currently, only a logarithmic approximation is
known \cite{hj-ospl-20}.

\paragraph{Applications.}

One natural application of this problem comes from machine learning.
Given a (single) point set $\P$ of size $n$ in $\Re^d$ and a
collection of features $f_1, \ldots, f_m$, where
$f_i : \Re^d \to \Re$, $f_i$ \emphw{distinguishes} between two points
$p$ and $q$ if $f_i(p)$ and $f_i(q)$ have different signs. Given a
collection of features $S \subseteq \{f_1, \ldots, f_m\}$ one can
assign each point $p$ a vector $v_S(p) \in \{-1,1\}^{\cardin{S}}$,
where each entry of $v_S(p)$ is the sign of a feature in $S$ evaluated
at $p$.  The point $v_S(p)$ is the \emphw{signature} of $p$ with
respect to $S$.  Consider the task of choosing a subset of features
$S \subseteq \{f_1, \ldots, f_m\}$, where $\cardin{S}$ is as small as
possible, such that for any $u \in \{-1,1\}^{\cardin{S}}$, the number
of points with the same signature as $u$ is at most $n/2$, formally
$\cardin{\Set{p \in \P}{u = v_S(p)}} \leq n/c$, with $c=2$.  The
choice $c=2$ in the last statement is arbitrary -- other values might
be desired, but if $c$ is too large, then there is no possible
solution. For example, $c=2^{\cardin{S}}$ is not feasible for
$\cardin{S} \gg d$, as an arrangement of $\cardin{S}$ hyperplanes in
$\Re^d$ has only $O(\cardin{S}^d)$ different cells (of various
dimensions).

Furthermore, one would like to apply this to several point sets
$\P_1, \ldots, \P_\ell$, where we would like to select the smallest
number of features $S$ such that for all
$u \in \{-1,1\}^{\cardin{S}}$,
$\cardin{\Set{p \in \P_i}{u = v_S(p)}} \leq \cardin{\P_i}/2$, for all
$i = 1, \ldots, \ell$.  A natural scenario for such an application is
in the realm of big data. Given a collection of (large) data sets, it
needs to be divided among different computers. The fewer the features
needed to get a split as described above, the faster one can decide
where to send such a point. Here, the required guarantee is that each
set gets reduced to at most half its size.

For the case where all the points have to be singletons in the induced
partition of features, this can be interpreted as a non-linear
dimension reduction of the input set into a hypercube, where the
dimension of the hypercube is as small as possible. Indeed, once we
picked a set $s$ of hyperplanes, $h_1, \ldots, h_s$, each one them has
an associated sign function $f_i(\pA) \in \{0, 1\}$, where a point
$\pA$ (not lying on any of the planes) has $f_i(\pa) = 1$ if $\pA$ is
on one side of $h_i$, and $0$ if $\pA$ is on the other size. This
naturally defines an embedding of $\P$ to the hypercube $\{0,1\}^s$,
as for all $ \pA \in \P$, we have
$F(\pA) = \bigl(f_1(\pA), \ldots, f_s(\pA) \bigr) \in \{0,1\}^s$.

\subsection{Background}

\paragraph{Ham sandwich theorem.}
The Ham-Sandwich Theorem is a well studied problem in both mathematics
and computer science. Since its inception, there have been many
results related to computing such cuts in higher dimensions
\cite{lms-ahsc-94}, as well as generalizations of the theorem
\cite{bhj-scsmh-08,bs-tcms-18,r-emdh-96,s-hsccts-19,st-gst-42}.  For
example, one such generalization is the following: Given well
separated convex bodies $C_1, \ldots, C_d$ in $\Re^d$ and constants
$\mu_i \in [0,1]$, there exists a unique hyperplane $h$ that contains
at least a $\mu_i$ fraction of the volume on the positive side $h^+$
for $i = 1,\ldots, d$ \cite{bhj-scsmh-08}. This result was then
extended to discrete point sets under certain conditions
\cite{sz-ghsc-10}. Notably, in this paper we consider the case when
the number of point sets can be much larger than the ambient dimension
$d$. The problem of simultaneously bisecting more than $d$ convex
bodies in $\Re^d$ using multiple hyperplanes has been studied
combinatorially \cite{bs-tcms-18,s-hsccts-19}, whereas our focus is on
the algorithmic aspects.

Other generalizations include the polynomial Ham-Sandwich Theorem, in
which one is interested in partitioning a point set using polynomials
rather than hyperplanes \cite{kms-spctd-12,st-gst-42}.  This
generalization, and the original Ham-Sandwich Theorem has a variety of
applications in geometric range searching
\cite{ams-rsss2-13,m-grs-94}.

\paragraph{Partial set cover.}

An instance of the \emph{set cover} problem is a pair $(\GSet, \HS )$,
where $\HS \subseteq 2^\GSet$. The problem is to compute a minimum
number $t$ of edges $f_1, \ldots, f_t \in \HS$ such that
$\cup_i f_i = \GSet$.

In the partial set cover problem, one is interested in covering at
least a certain fraction of the elements in a set system, using as few
sets as possible. Specifically, an instance of this problem is a tuple
$(\GSet, \HS, \alpha )$ (first two parameters are as in the set cover
problem, and $\alpha \in [0,1]$), and the problem is to compute the
minimum number $t$ of edges $f_1, \ldots, f_t \in \HS$, such that
$ \cardin{\cup_{i=1}^t f_i} \geq \alpha \cardin{\GSet}$.  For our
purposes, we need a parallel version of this problem (with many set
systems sharing sets, each with its own demand) to model our
problem. This variant is formally defined in \probref{pcms} below.

For the standard partial set cover problem, an
$O( \log n)$-approximation is well known, and follows from the greedy
algorithm (see below for details). In geometric settings, Inamdar and
Varadarajan \cite{iv-pcgss-18} showed that partial set cover can be
approximated to within $O( \beta)$, where $\beta$ is the approximation
ratio for the set cover version of the problem. Because many geometric
problems admit much better than $O(\log n)$-approximations, this
results in an improvement to the partial set cover version of the
problem. However, it is not clear how to apply their algorithm in the
parallel setting.

\subsection{Our results}

We reduce \probref{prob} to a generalized instance of \emph{partial
   set cover}, where we allow multiple ground sets, with different
demands, and show that the standard greedy algorithm for submodular
optimization can be applied to this problem.

\paragraph{Sketch of the greedy algorithm.}
To solve \probref{prob}, let $U = \cup_i \P_i$ and let $H$ be the
collection of all combinatorially different hyperplanes with respect
to $U$.  Consider the arrangement $\Arr = \ArrX{H}$ of $H$.  We
introduce an edge between a pair of points of $\P$ if they lie in the
same cell of $\Arr$. If we consider the process of adding the
hyperplanes from $H$ as an incremental process, then initially every
point is in the same cell as all the other points. Modeling this as a
graph, we start with a clique, and every hyperplane $h$ added
disconnects the edges which correspond to segments that $h$
intersects. In particular, a point is in a cell with at most $m$
points if it has degree $m-1$ in the remaining graph. As such, this
can be interpreted as a parallel version of set cover, where every
vertex induces its own instance, which requires a certain number of
edges adjacent to it to be covered (i.e., cut). Naturally, parallel
versions of set cover can be solved using a greedy algorithm that
picks the hyperplane that cuts the largest number of edges that still
need cutting (being somewhat informal). However, it is somewhat more
natural to describe the greedy algorithm using the framework of
submodular optimization.

\paragraph{Paper organization.}
In \secref{background} we provide the necessary background on
minimization under submodular constraint needed for our main
result. We then show how to solve the multiple partial set cover
problem in \secref{pcms}. Next, in \secref{demand}, we study the
problem of partitioning a set into smaller sets, such that each
element in each of the smaller sets meet a given demand
requirement. The final result, stated in \thmref{the:result}, provides
a logarithmic approximation for our problem by reducing it to the
aforementioned problems.

\section{Preliminaries}
\seclab{background}

For a set $X$, and an element $x$, let $X + x = X \cup \brc{ x}$, and
$X-x = X \setminus \brc{x}$. A \emphw{set system} is a pair
$(\GSet,\HS)$, with $\HS \subseteq 2^\GSet$. The set system
$(\GSet, \HS)$ can also be viewed as a hypergraph with the vertex set
$\GSet$, and the sets in $\HS$ as \emphi{edges}.

\subsection{Submodular minimization}

For the sake of completeness, we present the analysis of the greedy
algorithm for finding a minimal solution satisfying an integer valued
submodular constraint.  In this case, the task is to compute the
smallest set of edges that provides the same utility as using all the
edges available.

Let $(\GSet,\HS)$ be a given set system, and assume we have a monotone
function $f:2^\HS \rightarrow \ZZ$. Here a function is
\emphi{monotone} if $\HSB \subseteq \HSA \subseteq \HS$ implies that
$f(\HSB) \leq f(\HSA) \leq f(\HS)$.  Intuitively, the function
$f(\HSB)$ measures the \emphw{benefit} of a set $\HSB$ -- the higher
the value of $f$ is, the higher the benefit. In particular,
$\fmax = f(\HS)$ is the maximum benefit possible.

We also assume that $f$ is \emphw{submodular}, that is for any
$\he \in \HS$, and for all
$\HSB \subseteq \HSA \subseteq \HS \setminus \brc{ \he}$, we have that
\begin{align*}
 \qquad
  \margY{\he}{\HSB} = 
  f\pth{ \HSB + \he } - f(\HSB) \geq 
  f\pth{ \HSA + \he }  - f(\HSA)
  = \margY{\he}{\HSA}. 
\end{align*}
Submodularity is known in economics as \emphw{diminishing returns} --
the marginal benefit (per unit) of allocating more resources to solve
a problem decreases as more resources are allocated.

\begin{problem}
    \problab{max}%
    Under the above settings, the problem at hand is to compute (or
    approximate) the smallest (cardinality) set $\Opt \subseteq \HS$,
    such that $f(\Opt) = \fmax$.
\end{problem}

\begin{example}
    Consider an instance of set cover $(\GSet, \HS )$, with
    $n = \cardin{\GSet}$. Given a family $\HSB \subseteq \HS$ of
    edges, its \emph{benefit} is the number of elements in $\GSet$ the
    edges of $\HSB$ cover. That is,
    $f(\HSB) = \cardin{\cup_{\hz \in \HSB}^{} \hz}$. It is not hard to
    verify that $f$ is monotone and submodular. Solving \probref{max}
    here corresponds to computing a minimum set cover for $\GSet$.
\end{example}

Consider the greedy algorithm that starts with an empty solution
$\curr_0$. In the $i$\th iteration, the algorithm picks the edge
$\he_i'\in \HS$ that maximizes the value
$f(\curr_{i-1} +\he_i') - f(\curr_{i-1})$, and updates
$\curr_i = \curr_{i-1} + \he_i'$.  The algorithm stops when
$f(\curr_i) = \fmax = f\pth{ \HS}$.

\begin{theorem}[Wolsey \cite{w-agass-82}]
    \thmlab{wolsey}%
    Given a set system $(\GSet,\HS)$, and a non-negative monotone
    submodular function $f:2^\HS \rightarrow \ZZ$, the greedy
    algorithm, described above, outputs a solution with
    $O( \kopt \log \fmax )$ edges of $\HS$, where
    $\kopt = \cardin{\Opt}$ is the size of the smallest set
    $\Opt \subseteq \HS$ such that $f(\Opt) = \fmax = f( \HS)$.
\end{theorem}

\begin{proof}
    This result is by now classical, and we include the proof only for
    the sake of completeness.  Let
    $\Opt = \brc{\heo_1, \ldots, \heo_\kopt}$ be the optimal
    solution. Consider a current solution $\curr_i \subseteq \HS$ at
    iteration $i$, and observe that by monotonicity, we have
    \begin{equation*}
        \fmax = f(\Opt) \leq f(\curr_i \cup \Opt ) \leq \fmax.        
    \end{equation*}
    As such, we have $ f(\curr_i \cup \Opt ) = \fmax$.  Let
    $\Delta_i = f(\Opt) - f(\curr_i)$ be the \emphi{deficiency} of
    $\curr_i$.  For $j=0, \ldots, k$, let
    \begin{math}
        \HSS_j = \curr_i \cup \brc{ \heo_1, \ldots, \heo_j}.
    \end{math}
    Set $\delta_j = f\pth{ \HSS_j} - f\pth{ \HSS_{j-1}}$. We have that
    \begin{align*}
      \sum_{j=1}^k \delta_j%
      =%
      f\pth{\curr_i \cup \Opt} - f(\curr_i)%
      =%
      \fmax - f(\curr_i) = \Delta_i.
    \end{align*}
    Hence, there is an index $j$, such that
    $\delta_j \geq \Delta_i/k$. Now, by submodularity, we have that
    \begin{align*}
      f(\curr_i + \heo_j) - f(\curr_i)%
      \geq %
      f(\HSS_{j-1} + \heo_j) - f(\HSS_{j-1})%
      = %
      \delta_j
      \geq \Delta_i/k.
    \end{align*}
    However, the greedy algorithm adds an element $\he$ that maximizes
    the value of $\Delta_{\curr_i}(\he)$, which is at least
    $\Delta_i/k$. Put differently, the added element decreases the
    deficiency of the current solution by a factor $\leq
    1-1/\kopt$. Therefore the deficiency in the end of the $i$\th
    iteration is at most
    \begin{math}
        \Delta_i \leq (1-1/\kopt)^i \Delta_0 = (1-1/\kopt)^if(\Opt).
    \end{math}
    This quantity is less than one for $i = O( \kopt \log \fmax )$.
\end{proof}

\section{Problems and reductions}

\subsection{\PCMS: Partial cover for multiple sets}
\seclab{pcms}

\begin{problem}[\PCMS]
    \problab{pcms}%
    The input is a set system $(\USet,\HS)$, and a collection
    $\GSSet = \Set{\GSet_i \subseteq \USet}{i=1,\ldots,m}$ of ground
    sets, where the universe $\USet$ is of size $n$.  In addition,
    each ground set $\GSet_i$ has a \emphi{demand}, denoted by
    $\demandX{\GSet_i}$, which is a non-negative integer.  A valid
    solution for such an instance, is a collection
    $\HSA \subseteq \HS$, such that $\bigcup_{\hy \in \HSA} \hy$
    covers at least $\demandX{\GSet_i}$ elements of $\GSet_i$, for
    $i=1, \ldots, m$.
\end{problem}

\begin{remark}
    In the following, to simplify the exposition, we assume that the
    given instances being solved are feasible. Otherwise, the
    approximation algorithm would fail to generate a solution thus
    proving the unfeasibility of the given instance.
\end{remark}

\begin{lemma}
    \lemlab{pcms}%
    Let $\pth{ \USet, \GSSet, \HS}$ be an instance of \emph{partial
       cover of multiple sets} (\PCMS), where $n = \cardin{\USet}$,
    $\GSSet$ is a family of $m$ ground sets, and $\HS$ is a family of
    edges. Furthermore, each ground set of $\GSSet$ has an associated
    demand.  Then, the greedy algorithm computes, in polynomial time,
    an $O\bigl( \log (mn) \bigr)$-approximation to the minimal size
    set $\Opt \subseteq \HS$ that meets all the demands of the ground
    sets.
\end{lemma}
\begin{proof}
    Consider a partial solution $\curr \subseteq \HS$. The
    \emph{service} of $\curr$ to $\GSet_i$ is
    \begin{equation*}
        f_i (\curr)%
        =%
        \min\pth{\bigl.  \cardin{\GSet_i \cap \pth{ \cup\curr }} ,
           \,
           \demandX{\GSet_i}},        
    \end{equation*}
    where $\cup \curr = \cup^{}_{\he \in \curr}\, \he$. That is
    $f_i(\curr)$ is the number of elements of $\GSet_i$ the union of
    the edges of $\curr$ covers. Observe that $f_i(\emptyset) = 0$,
    $f_i$ is clearly monotone, and its maximal value is $n$. As for
    submodularity, consider sets $\HSB \subseteq \HSA \subseteq \HS$,
    and an edge $\he \in \HS$, and note that
    \begin{math}
        f_i\pth{ \HSB + \he} - f_i\pth{ \HSB }%
        \geq f_i\pth{ \HSA + \he} - f_i\pth{ \HSA },
    \end{math}
    as $\he$ potentially covers more new elements of $\GSet_i$ when
    added to a smaller cover. For the given \PCMS instance and a
    solution $\HSB \subseteq \HS$, the target function is
    \begin{align*}
      f(\HSB) = \sum_{i=1}^m f_i( \HSB).
    \end{align*}
    The function $f$ is a sum of submodular functions. As such, $f$ is
    submodular itself. Observe that $f(\HS) \leq mn $. Now, using the
    algorithm of \thmref{wolsey} implies the result.
\end{proof}

\begin{remark}
    One can obtain an $O(\log m)$-approximation for \probref{pcms} via
    LP rounding \cite{ky-aacpi-05}, which is useful when $m$ is much
    smaller than $n$.  However, this does not change our final result,
    since the number of ground sets in our reduction is polynomial in
    $n$ (see \lemref{single:set}).
\end{remark}

\subsection{Cutting a set into smaller pieces}
\seclab{demand}

We are given a set-system $(\GSet, \HS)$, where $n = \cardin{\GSet}$.
A set $\HSB \subseteq \HS$ of edges, induces a natural partition of
$\GSet$, where two elements $x, y \in \GSet$ are in the same set of
the partition $\iff$ $x$ and $y$ belong to the same set of edges in
    $\HSB$. Formally, $x \equiv y$ $\iff$ $\HSB \cap x = \HSB \cap y$,
        where $\HSB \cap x = \Set{ f \in \HSB}{x \in f}$.  The
        partition of $\GSet$ induced by $\HSB$ (i.e., the equivalence
        classes of $\equiv$) is the \emphi{arrangement} of $\HSB$,
        denoted by $\ArrX{\HSB}$. A set of $\ArrX{\HSB}$ is a
        \emphi{face} of $\ArrX{\HSB}$. For an element $x \in \GSet$,
        the face of $\ArrX{\HSB}$ that contains $x$ is denoted by
        $\faceY{x}{\HSB}$.

\begin{example}
    For $\GSet = \brc{1,2,3,4,5}$, and
    $\HSB = \bigl\{\Bigl. \brc{ 1,2, 3}, \brc{3,4,5}\bigr\}$, we have
    \begin{equation*}        
        A(\HSB) = \bigl\{ \Bigl. \brc{ 1,2}, \brc{3}, \brc{4,5} \bigr\}.
    \end{equation*}
\end{example}

\begin{problem}[Reduce by half]
    \problab{r:b:h}%
    Given a set system $(\GSet, \HS)$, with $n = \cardin{\GSet}$, find
    a minimum sized set $\HSB \subseteq \HS$ such that every face of
    $A(\HSB)$ is of size at most $n/2$.
\end{problem}

\begin{problem}[\PTD: Partition to demand]
    \problab{ptd}%
    Given a set system $(\GSet, \HS)$, where $n = \cardin{\GSet}$, and
    an integral \emphi{demand} $\demandX{v} \geq 0$, for each
    $v \in \GSet$, find a minimum sized set $\HSB \subseteq \HS$, such
    that for every $v\in \GSet$,
    $\cardin{\faceY{v}{\HSB}} \leq \demandX{v}$.
\end{problem}

Observe that \probref{r:b:h} can be reduced to \probref{ptd} by
setting the demand of every vertex in the ground set to $n/2$.

\begin{lemma}
    \lemlab{single:set}%
    Given an instance $(\GSet, \HS)$ of \PTD, with
    $n = \cardin{\GSet}$, there is a greedy algorithm that computes,
    in polynomial time, an $O( \log n)$-approximation to the optimal
    solution.
\end{lemma}
\begin{proof}
    Consider the complete graph $K_n=(\GSet,E)$, where
    $E = \Set{xy}{x, y \in G}$. For every element $x \in \GSet$,
    consider the associated cut $E_x = \Set{xy}{y \in \GSet - x}$. A
    set $\he \in \HS$ \emph{cuts} $xy$ if
    $\cardin{\he \cap \brc{x,y}} = 1$. In particular, let
    $\cutX{\he} = \Set{ xy }{x \in \he, y \in \GSet \setminus \he}$ be
    the set of edges of $K_n$ that $\he$ cuts.

    Now, a set of edges $\HSA \subseteq \HS$ meets the demand of
    $v \in \GSet$, if the edges of $\HSA$ cut at least
    $n - \demandX{v}$ edges of $E_v$ (e.g., if $\demandX{v}=n-1$, hen
    one needs to cut one edge attached to $v$). Put differently, the
    partial cover $\bigcup_{\he \in \HSA} \cutX{\he}$ covers at least
    $n - \demandX{v}$ edges of $E_v$. Thus, let $\USet' = E$ be the
    universe set, and
    \begin{math}
        \GSSet' = \Set{ E_v }{ v \in \GSet}
    \end{math}
    be the set of ground sets. Here a ground set $E_v \in \GSSet'$ has
    demand $\demandX{E_v} = n - \demandX{v}$. The family of allowable
    sets to be used in the cover is
    \begin{math}
        \HS' = \Set{\cutX{\he} }{ \he \in \HS}.
    \end{math}

    The triple $\pth{ \USet', \GSSet', \HS'}$ is an instance of \PCMS,
    with $n' = |\USet'| = O(n^2)$ and $m' = \cardin{\GSSet'} = n$.
    The greedy algorithm yields an
    $O\bigl( \log (n'm') \bigr)$-approximation in this case, by
    \lemref{pcms}. As $\log (n'm') = O( \log n)$, the claim follows.
\end{proof}

\subsection{Cutting a Ham-Sandwich into small pieces}

\begin{problem}[\RMC: Reduce measures via cuts]
    \problab{rmc} The input is a triplet $(\USet,\GSSet, \HS)$ with
    $n = \cardin{\USet}$. Here
    $\GSSet = \brc{ \GSet_1, \ldots, \GSet_m}$ is a collection of
    ground sets that are not necessarily disjoint, and $\HS$ is a
    collection of edges.  For every ground set $\GSet_i$, there is an
    associated target size $\mu_i \leq \cardin{\GSet_i}$. The problem
    is to compute a minimal set $\Opt \subseteq \HS$, such that, for
    all $i$, and any cell $\cell$ of $\ArrX{\Opt}$, we have
    \begin{math}
        \cardin{ \cell \cap \GSet_i} \leq \mu_i.
    \end{math}
\end{problem}

\begin{lemma}
    \lemlab{smaller:pieces}%
    Given a feasible instance $(\USet,\GSSet, \HS)$ of \RMC with
    $n = \cardin{\USet}$ and $m = \cardin{\GSSet}$, one can compute,
    in polynomial time, an $O( \log (nm ))$-approximation to the
    smallest set $\Opt \subseteq \HS$ that satisfies the given
    instance.
\end{lemma}
\begin{proof}
    For a set $\GSet_i \in \GSSet$, and an element $v \in \USet$, let
    $\demandY{i}{v} = \mu_i$ if $v \in \GSet_i$, and otherwise
    $\demandY{i}{v} = n$. The pair $(\USet, \HS)$ with the demand
    function $\demandY{i}{\cdot}$ form an instance of \PTD
    (\probref{ptd}), and its approximation algorithm
    \lemref{single:set} has an associated submodular function
    $f_i (\cdot)$, that is non-negative, monotone, submodular and has
    maximum value $n^2$.

    Consider the submodular function $f = \sum_i f_i$, and let
    $\fmax = f( \HS)$. Clearly, $f$ is submodular, monotone, and has
    maximum value $mn^2$. Furthermore, a subset $\HSA \subseteq \HS$
    such that $f(\HSA) = \fmax$ is a valid solution to the given
    instance. As such, one can plug this into the algorithm of
    \thmref{wolsey} and get the desired approximation.
\end{proof}

With all of the ingredients assembled, we are ready to tackle
\probref{prob}.

\begin{theorem}
    \thmlab{the:result}%
    Let $\P_1, \ldots, \P_m$ be $m$ (not necessarily disjoint) point
    sets in $\Re^d$, where $n = \sum_i \cardin{ \P_i}$.  For each
    point set $\P_i$, we are given an integer parameter
    $0 < \mu_i \leq \cardin{\P_i}$.  The task at hand is to compute a
    minimal set of hyperplanes $H$ such that for every cell $\cell$ in
    the arrangement $\ArrX{H}$, $\cell$ contains at most $\mu_i$
    points of $\P_i$, for all $i = 1, \ldots, m$.  One can
    $O\bigl(\log(mn)\bigr)$-approximate, in $O(m n^{d+3})$ time, the
    optimal solution.
\end{theorem}
\begin{proof}
    The reduction is straightforward and uses \lemref{smaller:pieces}.
    Let the shared ground set be $\USet = \cup_i \P_i$. Let $\GSSet$
    be the family of ground sets
    $\Set{\GSet_i = \P_i}{i = 1, \ldots, m}$. Finally, let $H$ be the
    (finite) number of combinatorially different hyperplanes with
    respect to $\USet$. For each $h \in H$, let $h^+$ be one of the
    two halfspaces bounded by $h$ (which halfspace is not important --
    taking the other one corresponds to ``flipping'' the corresponding
    coordinate of the signature induced the arrangement).  Add the
    set $\Set{p \in U}{p \in h^+}$ to the collection of subsets
    $\HS$. The values $\mu_i$ remain unchanged.  This forms an
    instance of \probref{rmc}, and thus we can apply
    \lemref{smaller:pieces} to obtain the desired separating
    hyperplanes.

    As for the running time, computing the set system takes
    $O(n^{d+2})$ time by brute force. Indeed, unraveling the above
    reduction, the shared ground set is made of $\binom{n}{2}$ pairs
    of points of $\USet$. Every point has up to $m$ different sets of
    such pairs that needs to be partially covered. Fortunately, there
    are only $O(n^d)$ edges in the resulting set system. Evaluating
    the contribution of a new edge (in the set system) to the target
    function takes $O(n^2m)$ time. As there are $O(n^d)$ edges in set
    system, it follows that evaluating all edges takes $O(n^{d+2} m)$
    time. Finally, it is easy to verify that the algorithm performs at
    most $n$ iterations.
\end{proof}

    No effort was made to improve the running time of the algorithm of
    \thmref{the:result}.

\section{Open problems}

The most natural open problem is to try and further improve the
approximation quality of \thmref{the:result}. The same applies to all
the other problems here, which potentially might have better
approximation ratios because of the underlying geometry. On the other
hand, it would be interesting to prove (conditional) lower bounds on
the hardness of approximation of these problems.

\AlgVer{\bmhead{Acknowledgments}}%

\paragraph*{Acknowledgments} %
\AlgVer{%
   Sariel Har-Peled was partially supported by NSF AF awards
   CCF-1421231, CCF-1217462, and CCF-1907400. Mitchell Jones was
   partially supported by NSF AF awards CCF-1421231 and CCF-1907400.

}%
The authors also thank the anonymous referees for their detailed and
useful feedback.

\BibTexMode{%
   \bibliographystyle{alpha}%
   \bibliography{split_splat}%
}%

\BibLatexMode{\printbibliography}

\end{document}